\documentclass[10pt,journal]{IEEEtran}

\usepackage{amsmath,amssymb,amsfonts}
\usepackage{cite}
\usepackage{graphicx}
\usepackage{booktabs}
\usepackage{array}
\usepackage{url}
\usepackage{balance}
\usepackage[caption=false,font=footnotesize]{subfig}
\usepackage{microtype}
\usepackage{enumitem}
\usepackage{multirow}

\graphicspath{{./}}

\newcommand{\cV}{\mathcal{V}}
\newcommand{\cE}{\mathcal{E}}
\newcommand{\cJ}{\mathcal{J}}
\newcommand{\cT}{\mathcal{T}}

\newcommand{\bbE}{\mathbb{E}}

\newcommand{\R}{\mathsf{R}}
\newcommand{\Rs}{\mathsf{R}_{S}}
\newcommand{\D}{\mathsf{D}}
\newcommand{\B}{\mathsf{B}}

\newtheorem{prop}{Proposition}

\newtheorem{cor}{Corollary}

\usepackage{tikz}
\usetikzlibrary{arrows.meta, positioning,calc,fit,backgrounds}
\usepackage[hidelinks]{hyperref}
\hypersetup{breaklinks=true}
\usepackage{xcolor}

\begin{document}

\title{Capacity-Aware Sparse Backpropagation for\\In-Network Learning with Finite-Rate Messages}

\author{Mohammad Reza Deylam Salehi and Hakima Chaouchi
\thanks{M. R. Deylam Salehi is a research engineer at T{\'e}l{\'e}com SudParis, Palaiseau, France (email: mdeylamsalehi@telecom-sudparis.eu), and H. Chaouchi is a professor at Télécom SudParis/Institut Mines-Télécom/Institut Polytechnique de Paris,  Palaiseau, France, (email: hakima.chaouchi@imt.fr). 
Founded by the French government under the EU Important Project of Common European Interest on Microelectronics and Communication Technologies (IPCEI ME/CT).}
}

\maketitle

\begin{abstract}

In-network learning (INL) trains distributed neural modules by exchanging activations and backpropagated errors over network links. We propose a two-stage approach that selects a capacity-aware shortest-path tree for sparse training and controls the information rate of stochastic latent messages using a Gaussian rate proxy. We derive a bounded-loss generalization bound that accounts for topology-induced distortion, message rate, and data dependence. Across two network topologies, the proposed method removes up to $77\%$ of candidate training links and reduces the capacity-weighted communication-cost proxy by up to $92.9\%$. Experiments on distributed MNIST and subject-held-out MHEALTH data further evaluate the trade-off between communication cost and classification accuracy.
\end{abstract}



\section{Introduction}
\label{sec:intro}

In mobile edge and cyber-physical systems, measurements relevant to distributed prediction tasks are collected by multiple devices. For example, wearable inertial sensors may collect complementary observations for activity recognition. A nearby edge gateway can then fuse compressed features without collecting all raw sensor streams. In-network learning (INL) addresses this setting by placing neural modules at network nodes. Intermediate nodes fuse latent representations during inference and propagate error signals during training~\cite{moldoveanu2023}. Compared with federated learning (FL)~\cite{mcmahan2017}, INL preserves the distributed nature of both training and inference. It also differs from split learning~\cite{moldoveanu2023} by allowing multiple distributed observations to be fused within the network.

During training, backpropagation follows the communication graph, so every active communication link carries both forward activations and reverse error signals. Consequently, the communication cost grows with the number and capacity of the active links. Removing a training link suppresses its forward activation and backward error exchange; it does not, by itself, reduce the size of a transmitted scalar or guarantee lower end-to-end delay. Thus, dense INL scales exchange with the number and width of active links.

Communication-efficient training has been studied via gradient compression~\cite{lin2017deep}, task-oriented feature coding~\cite{shao2023task}, and joint pruning at the network edge~\cite{jankowski2020joint, khalesi2026distributed}. These works focus on centralized or split-inference settings. For fully distributed inference graphs, structured topology selection for backpropagation has not been explicitly addressed in SoA.

This letter combines two complementary elements. First, a capacity-aware Dijkstra SPT~\cite{dijkstra1959} retains low-cost inference and backpropagation routes from data nodes to the fusion node. Second, motivated by finite-rate gating in mixture-of-experts (MoE) models~\cite{khalesi2026,salehi2026expert} and the information bottleneck principle~\cite{tishby2000information}, a finite-rate stochastic bottleneck controls the information carried by the selected messages.

\noindent\textbf{Contributions:} We use shortest-path pruning to select the forward and backward INL dependency graph, followed by stochastic-message training with a Gaussian KL penalty. We define the limits of the routing guarantee, distinguish network metrics, and evaluate the resulting trade-offs in controlled networks, MNIST, and independent wearable recordings.

\begin{figure*}[t]
\centering
\begin{tikzpicture}[
  font=\scriptsize,
  block/.style={draw,rounded corners=1.5pt,thick,align=center,
                minimum height=6mm,fill=gray!7},
  gate/.style={draw,rounded corners=1.5pt,thick,align=center,
               minimum height=8mm,fill=blue!5},
  forward/.style={-{Stealth[length=1.8mm]},very thick,black},
  backward/.style={-{Stealth[length=1.7mm]},thick,dashed,red!80!black}
]

\node[block] (x1) at (0, 0.95) {$X_1$};
\node[block] (x2) at (0, 0.0) {$X_2$};
\node at (0,-0.47) {$\vdots$};
\node[block] (xJ) at (0,-0.95) {$X_J$};

\node[block] (e1) at (2.1, 0.95) {encoder $f_1$};
\node[block] (e2) at (2.1, 0.0)  {encoder $f_2$};
\node at (2.1,-0.47) {$\vdots$};
\node[block] (eJ) at (2.1,-0.95) {encoder $f_J$};

\node[gate] (b1) at (5.1, 0.95) {$q_{\phi_1}(A_1\mid Z_1)$};
\node[gate] (b2) at (5.1, 0.0)  {$q_{\phi_2}(A_2\mid Z_2)$};
\node at (5.1,-0.47) {$\vdots$};
\node[gate] (bJ) at (5.1,-0.95) {$q_{\phi_J}(A_J\mid Z_J)$};

\node[block] (r1) at (8.45, 0.62) {relay $R_1$\\mean/sum};
\node at (8.45,0.09) {$\vdots$};
\node[block] (r2) at (8.45,-0.67) {relay $R_2$\\mean/sum};

\node[block] (fusion) at (12,0) {fusion\\$q_{W_r}(Y\mid\widetilde M_r)$};

\draw[forward] (x1) -- (e1);
\draw[forward] (x2) -- (e2);
\draw[forward] (xJ) -- (eJ);

\draw[forward] (e1) -- node[below] {$Z_1$} (b1);
\draw[forward] (e2) -- node[below=3pt] {$Z_2$} (b2);
\draw[forward] (eJ) -- node[below=3pt] {$Z_J$} (bJ);

\draw[backward] (b1.west) to[bend right=23] (e1.east);
\draw[backward] (b2.west) to[bend right=21] (e2.east);
\draw[backward] (bJ.west) to[bend right=23] (eJ.east);

\draw[forward] (b1) -- node[below, near start] {$A_1$} (r1);
\draw[forward] (b2) -- node[below=3pt, near start] {$A_2$} (r1);
\draw[forward] (bJ) -- node[below, near start] {$A_J$} (r2);

\draw[backward] (r1) to[bend right=9] (b1);
\draw[backward] (r1) to[bend right=9] (b2);
\draw[backward] (r2) to[bend right=9] (bJ);

\draw[forward] (r1) -- node[below] {$M_{R_1}$}(fusion);
\draw[forward] (r2) -- node[below] {$M_{R_2}$} (fusion);

\draw[backward] (fusion) to[bend right=10] (r1);
\draw[backward] (fusion) to[bend right=10] (r2);

\node[anchor=west,font=\scriptsize,align=left] at (0,-1.7)
{\textcolor{black}{\rule[.3ex]{5mm}{0.7pt}} forward activations \quad
 \textcolor{red!80!black}{- - -} reverse error messages / local backprop};

\node[draw,dashed,rounded corners=2pt,fit=(x1)(xJ)(fusion),inner sep=1.95mm,
  label={[font=\scriptsize]above:Representative Dijkstra-pruned in-network learning paths}] {};

\end{tikzpicture}
\caption{System model of D-INL. Each sensor encoder $f_j$ produces a latent $Z_j$ and a local finite-rate Gaussian bottleneck yields $A_j$. The Dijkstra SPT routes selected messages to relays and the fusion node; reverse error propagation uses the same tree links.}
\label{fig:system_model_tikz}
\vspace{-2mm}
\end{figure*}

\section{System Model and Sparse Backpropagation}
\label{sec:sys-mod}

Fig.~\ref{fig:system_model_tikz} summarizes the setting. Let $G=(\cV,\cE,C)$ be a directed acyclic graph with nodes $\cV=\{1,\ldots,N\}$, edge capacities $C_e$, and fusion node $r$. A subset $\cJ\subset\cV$ observes distributed data $X_j$, $j\in\cJ$, and the fusion node predicts $Y$. Each non-fusion node $v$ contains a neural module $f_v(\cdot;W_v)$. For a selected training subgraph $\cT\subseteq\cE$, let $Z_v$ be the local latent state and $M_v$ the transmitted message:
\begin{align}
\label{eq:latent}
Z_v=f_v\big(X_v,\{M_u:(u,v)\in\cT\};W_v\big)\ ,
\end{align}
where $X_v=\emptyset$ at relays; sensors send $M_j=A_j\sim q_{\phi_j}(\cdot\mid Z_j)$ and relays send $M_v=Z_v$. Each message has dimension $d$; relay mean/sum aggregation keeps the module input width fixed after pruning. Fusion uses $q_{W_r}(Y\mid\widetilde M_r)$, where $\widetilde M_r$ concatenates source/relay slots, zero-padded for inactive links. We assume bounded task loss $\ell\in[0,1]$. The experiments instead optimize cross-entropy, to which the bounded-loss inequality does not directly apply. For a training set $S$ of size $m$,
\begin{align}
\R(W_\cT)=\bbE[\ell(\hat Y_\cT,Y)],\: \:
\Rs(W_\cT)=\frac{1}{m}\sum_{i=1}^{m}\bbE[\ell(\hat y_{\cT,i},y_i)],
\end{align}
where expectations include stochastic gates or bottlenecks. The per-epoch training exchange for $q$ samples is defined as
\begin{align}
\B_\cT=2sq\sum_{e\in\cT} d_e\ ,
\label{eq:bandwidth}
\end{align}
where $d_e$ is the message width on edge $e$, $s$ is the bit-width per scalar, and the factor of two accounts for forward activations and backward errors. This is a nominal payload count, not a physical-link measurement.

Removing an edge $(u,v)$ eliminates both the forward message $M_u\to v$ and the reverse error signal $v\to u$, avoiding actual network transmissions. The subgraph must nonetheless preserve a path from every data node to $r$, making the design a constrained topology-selection problem.

A pure bit count~\eqref{eq:bandwidth} ignores that low-capacity links are slow even when they carry few scalars. We therefore also track a capacity-weighted transmission cost as follows
\begin{align}
\label{eq:latency}
\Lambda_\cT=2\,s\sum_{e\in\cT}\frac{d_e}{C_e+\varepsilon}\ ,
\end{align}
which counts each \emph{active edge once}. This is an evaluation proxy, not measured energy, delay, or network completion time, and need not equal the path objective optimized by Dijkstra.

\subsection{Dijkstra-Pruned Topology}

Each directed edge $e=(u,v)$ is assigned the additive cost
\begin{align}
\omega_e=\alpha\frac{s d_e}{C_e+\varepsilon}+\beta\tau_e+\gamma(1-\rho_e)\ ,
\label{eq:edgecost}
\end{align}
where $C_e$ is an effective payload-rate attribute (or a normalized capacity), $\tau_e$ is a latency attribute, $\rho_e\in[0,1]$ is reliability, and $\alpha,\beta,\gamma\ge0$. We use $\varepsilon=10^{-8}$; experiments isolate capacity-aware routing with $(\alpha,\beta,\gamma)=(1,0,0)$. Running Dijkstra's algorithm~\cite{dijkstra1959} on the reverse graph from $r$ gives the minimum-cost path from each data node to $r$. Their union defines $\cT_D$. Training is then performed only on $\cT_D$: forward activations travel on tree edges and error vectors propagate backward on the reverse tree. For equal message widths,
\begin{align}
\frac{\B_{\cT_D}}{\B_{\rm full}}=\frac{|\cT_D|}{|\cE_{\rm tr}|}\ .
\label{eq:reduction}
\end{align}
The SPT computation costs $O(|\cE|\log|\cV|)$ and is recomputed only when link states change, decoupling network-layer topology selection from neural optimization.

\begin{figure}[t]
\centering
\begin{tikzpicture}[
    font=\scriptsize,
    sensor/.style={circle, draw, thick, minimum size=5.5mm, fill=blue!8},
    relay/.style={circle, draw, thick, minimum size=6mm, fill=orange!12},
    fusion/.style={circle, draw, very thick, minimum size=7mm, fill=green!12},
    cand/.style={-{Stealth[length=1.6mm]}, dashed, gray!65, line width=0.35pt},
    tree/.style={-{Stealth[length=1.8mm]}, black, very thick},
    labeltxt/.style={font=\scriptsize}
]

\node[sensor] (s1) at (0,  1.25) {$S_1$};
\node[sensor] (s2) at (0,  0.75) {$S_2$};
\node[sensor] (s3) at (0,  0.25) {$S_3$};
\node[sensor] (s4) at (0, -0.25) {$S_4$};
\node[sensor] (s5) at (0, -0.75) {$S_5$};
\node[sensor] (s6) at (0, -1.25) {$S_6$};

\node[relay] (r1) at (2.1,  0.85) {$R_1$};
\node[relay] (r2) at (2.1,  0.00) {$R_2$};
\node[relay] (r3) at (2.1, -0.85) {$R_3$};

\node[fusion] (f) at (4.2,0) {$F$};

\node[labeltxt] at (0,1.75) {sensors};
\node[labeltxt] at (2.1,1.75) {relays};
\node[labeltxt] at (4.2,1.75) {fusion};

\foreach \s in {s1,s2,s3,s4,s5,s6}{
    \foreach \r in {r1,r2,r3}{
        \draw[cand] (\s) -- (\r);
    }
}

\foreach \s in {s1,s2,s3,s4,s5,s6}{
    \draw[cand] (\s) -- (f);
}

\foreach \r in {r1,r2,r3}{
    \draw[cand] (\r) -- (f);
}

\draw[tree] (s1) -- (r1);
\draw[tree] (s2) -- (r1);
\draw[tree] (s3) -- (r2);
\draw[tree] (s4) -- (r1);
\draw[tree] (s5) -- (r1);
\draw[tree] (s6) -- (r2);
\draw[tree] (r1) -- (f);
\draw[tree] (r2) -- (f);

\draw[cand] (0.2,-1.85) -- ++(0.55,0);
\node[anchor=west] at (0.72,-1.85) {candidate links};

\draw[tree] (2.4,-1.85) -- ++(0.55,0);
\node[anchor=west] at (2.9,-1.85) {SPT links};

\end{tikzpicture}
\caption{Candidate INL topology and Dijkstra SPT. Dashed edges are feasible training links; solid edges are retained for forward activation and backward error propagation.}
\label{fig:topology}
\end{figure}

\begin{prop}[SPT routing optimality]
\label{prop-spt-rout}
Assume $\omega_e\ge 0$ for all $e\in\cE$, and suppose each data node $j\in\cJ$ sends information to the fusion node through a single directed path $P_j$. The reverse-Dijkstra construction produces paths $\{P_j^D\}_{j\in\cJ}$ satisfying
\begin{align}
\label{eq-sum-costs}
\sum_{j\in\cJ}\sum_{e\in P_j^D}\omega_e=\sum_{j\in\cJ} d_\omega(j,r)
\le \sum_{j\in\cJ}\sum_{e\in P_j}\omega_e \ ,
\end{align}
for any feasible single-path routing $\{P_j\}$. Moreover, the union of minimum-cost paths, $\cT_D=\cup_{j\in\cJ}P_j^D$ preserves at least one information path from every data node to $r$.
\end{prop}

\begin{IEEEproof}
\label{proof-prop-spt-rout}
For non-negative weights, Dijkstra's algorithm on the reverse graph computes, for every node $j$, the exact minimum additive path cost $d_\omega(j,r)=\min_{P_j}\sum_{e\in P_j}\omega_e$ to the root $r$, together with a corresponding shortest path $P_j^D$. Hence for each source $j\in\cJ$ and any feasible single path $P_j$,
\begin{align}
\sum_{e\in P_j^D}\omega_e = d_\omega(j,r) \le \sum_{e\in P_j}\omega_e\ .
\label{eq:persource}
\end{align}
Since the per-source cost functional $J({P_j})=\sum_{j\in\cJ}\sum_{e\in P_j}\omega_e$ decomposes additively over the sources, summing~\eqref{eq:persource} over $j\in\cJ$ proves the desired inequality. Moreover, each $P_j^D$ terminates at $r$, implying that $\cT_D=\cup_{j\in\cJ}P_j^D$ contains a directed $j\to r$ path from every data node to $r$. Hence, $\cT_D$ is connected.
\end{IEEEproof}

Proposition~\ref{prop-spt-rout} concerns the \emph{source-weighted path sum} $J=\sum_j\sum_{e\in P_j}\omega_e$: a shared edge is counted for each source traversing it. By contrast, $\Lambda_\cT$ in~\eqref{eq:latency} counts each edge in the union once. Therefore, Dijkstra does \emph{not} globally minimize $\Lambda_\cT$, the Steiner union-edge objective, or multipath cost. The single-path constraint favors simpler scheduling but does not provide backup paths or relay-load guarantees.
\section{Finite-Rate Topology Trade-off}
\label{sec:fin-rat}

Shortest-path pruning alone may discard predictive information if high-capacity links coincide with informationally rich paths. Let $Z=(Z_j)_{j\in\cJ}$ and $A=(A_j)_{j\in\cJ}$ collect the sensor latents and stochastic messages, with information rate $R_g=I(Z;A)$. Here, $A$ denotes the collection of sensor bottleneck outputs; relay messages are deterministic functions of $A$ and $\cT$.
For fixed topology $\cT$, define
\begin{align}
\D_{\cT}(R_g)=\inf_{P(A|Z):\,I(Z;A)\le R_g}\bbE\big[\ell(\hat Y_{\cT,A},Y)\big].
\label{eq:rd}
\end{align}
This is the best achievable prediction loss when selected messages transmit at most $R_g$ nats per sample~\cite{blahut1972,khalesi2026}. The implementation is \emph{two-stage}: link attributes fix $\cT_D$ through~\eqref{eq:edgecost}; only then are weights trained by
\begin{align}
\widehat W_{\cT_D}=\arg\min_W\;\Rs(W_{\cT_D})+\lambda \widehat R_g(W_{\cT_D})\ ,
\label{eq:objective}
\end{align}
where $\widehat R_g$ is the sample-mean sum of sensor Gaussian KL terms to unit-normal priors, a variational upper bound on $I(Z;A)$. The routing step does \emph{not} minimize $\B_\cT$ or jointly optimize $W$ and $\cT$.

\begin{prop}[Rate-distortion-generalization bound]
\label{prop-risk}
Let $W_\cT$ be learned from $S$ on sparse topology $\cT$. Suppose $\bbE[\Rs(W_\cT)]\le \D_\cT(R_g)+\delta_m$, where $\delta_m\ge 0$ captures the optimization and empirical rate-distortion mismatch. Then
\begin{align}
\label{eq:bound}
\bbE[\R(W_\cT)]\le \D_\cT(R_g)+\delta_m+\sqrt{\frac{2}{m}I(S;W_\cT)}\ .
\end{align}
\end{prop}
\begin{IEEEproof}
\label{proof-eqband}
Write $a=\bbE[\R(W_\cT)]$ and $b=\bbE[\Rs(W_\cT)]$. The elementary bound $a\le b+|a-b|$ together with the hypothesis $b\le\D_\cT(R_g)+\delta_m$ gives
\begin{align}
\bbE[\R(W_\cT)]\le \D_\cT(R_g)+\delta_m+\big|\bbE[\R(W_\cT)]-\bbE[\Rs(W_\cT)]\big|.
\label{eq:split}
\end{align}
The last term in (\ref{eq:split}) is the expected generalization gap. Because the stochastic gate $A$ is internal to the model, the test loss $L=\ell(\hat Y_{\cT,A},Y)$ depends on the data only through the learned weights, i.e., $S\to W_\cT\to L$ is a Markov chain. For a loss bounded in $[0,1]$ (since, $\tfrac12$-sub-Gaussian), the mutual-information generalization inequality~\cite{xu2017} bounds the gap by $\sqrt{2I(S;W_\cT)/m}$. Substituting into~\eqref{eq:split} yields~\eqref{eq:bound}.
\end{IEEEproof}

\begin{cor}[Near-saturation regime]
\label{cor-saturation}
Assume $R_g\le C_g$ and $\D_\cT(R)$ is $L_\cT$-Lipschitz on $[C_g-\eta,C_g]$. If $0\le C_g-R_g\le\eta$, then
\begin{align}
\label{eq:satbound}
\bbE[\R(W_\cT)]\le \D_\cT(C_g)+L_\cT\eta+\delta_m+\sqrt{\frac{2}{m}I(S;W_\cT)}\ .
\end{align}
\end{cor}
\begin{IEEEproof}
The rate-distortion function $\D_\cT(R)$ is non-increasing in $R$, so $R_g\le C_g$ gives $\D_\cT(R_g)\ge\D_\cT(C_g)$. Within $[C_g-\eta,C_g]$ the Lipschitz property yields $\D_\cT(R_g)-\D_\cT(C_g)\le L_\cT(C_g-R_g)\le L_\cT\eta$, hence $\D_\cT(R_g)\le\D_\cT(C_g)+L_\cT\eta$. Substituting this into Proposition~\ref{prop-risk} gives~\eqref{eq:satbound}.
\end{IEEEproof}

The bound separates topology distortion $\D_\cT(R_g)$, optimization gap $\delta_m$, and data dependence $I(S;W_\cT)$. The coefficient $\lambda$ appears only in the \emph{training objective}, not in the bound. The theorem concerns bounded loss; unbounded NLL below is reported only as an empirical metric. Neither pruning nor the Gaussian KL proxy is claimed to prove a reduction in $I(S;W_\cT)$. In Corollary~\ref{cor-saturation}, $C_g$ is an \emph{information budget} (nats/sample), distinct from physical link capacities $C_e$.

\noindent\textbf{D-INL procedure:} Estimate link attributes $(C_e,\tau_e,\rho_e,d_e)$; compute~\eqref{eq:edgecost}; run reverse-Dijkstra from $r$; form the union $\cT_D$; and train neural weights using~\eqref{eq:objective} without changing the tree. A new SPT is required when network conditions change.

Since each non-fusion node has a unique next hop in $\cT_D$, activation forwarding and error propagation reduce to two pipelined waves, eliminating gradient synchronization over all feasible outgoing links.

\subsection{Communication and Complexity Scaling}

Let $c_v$ be the local per-sample compute cost of node $v$. Pruning does not change $\sum_v c_v$ for nodes retained in the tree; it changes only the inter-node exchange. With uniform message width $d$, the training-exchange gain is
\begin{align}
G_\B=1-\frac{\B_{\cT_D}}{\B_{\rm full}}=1-\frac{|\cT_D|}{|\cE_{\rm tr}|}\ ,
\label{eq:gain}
\end{align}
and the analogous transmission-cost gain follows from~\eqref{eq:latency} after weighting each edge by $1/(C_e+\varepsilon)$. The method therefore targets the networking bottleneck rather than the arithmetic cost of the neural modules. The coefficients in~\eqref{eq:edgecost} express a routing policy; latency and reliability weights are not experimentally validated here. The graph step runs in $O(|\cE|\log|\cV|)$ and is invoked only when link states change; since this is far less frequent than the inner stochastic-gradient loop, its amortized overhead is negligible.

\section{Numerical Results}
\label{numeric}

\subsection{Experimental Setup}

We evaluate D-INL on two distributed binary-classification topologies. The sensor--relay--fusion abstraction follows the directed INL model of Moldoveanu and Zaidi~\cite{moldoveanu2023}; our edge sets and capacities are controlled test configurations, not copies of a physical deployment. In both tasks, sensors observe noisy rotated/scaled two-dimensional projections of a common latent vector, relays aggregate incoming messages, and a two-layer fusion classifier predicts the label. All schemes use identical module dimensions and initializations per seed; results are averaged over $10$ seeds with validation-NLL checkpoint selection. Test predictions use the learned Gaussian posterior means. Evaluation NLL is unbounded and is not the loss assumed in Proposition~\ref{prop-risk}.

\noindent\textbf{Task~1 (small topology):} Six sensors, three relays, one fusion node, and $27$ candidate directed edges (Fig.~\ref{fig:topology}); $120$ training, $120$ validation, $1{,}000$ test samples. Capacity-aware reverse Dijkstra selects the $8$-edge tree $(0,6),(1,6),(2,7),(3,6),(4,6),(5,7),(6,9),(7,9)$.

\noindent\textbf{Task~2 (larger heterogeneous topology):} Ten sensors, five relays, one fusion node, and $65$ candidate edges; base link draws $U[0.5,2]$ scaled by $10$ for sensor--relay/relay--fusion links and $2$ for direct sensor--fusion links; $500$ training, $200$ validation, $2{,}000$ test samples.

\noindent\textbf{Baselines:} \emph{Dense INL} (all edges); \emph{Hop-count SPT} (capacity-blind, $\omega_e=1$); \emph{Random tree} (random connected spanning tree, re-sampled per seed); \emph{D-INL} (capacity-aware SPT, $\lambda=0$); and \emph{D-INL+rate} (capacity-aware SPT with Gaussian bottleneck, $\lambda=3\times10^{-3}$). In the synthetic studies, every method uses stochastic Gaussian encoders during training, including $\lambda=0$, so their KL diagnostics are comparable. The reported rate is the average encoder KL to a unit-Gaussian prior (nats/sample), which upper-bounds message mutual information under the specified encoder/prior; it is not a coded link bitrate~\cite{tishby2000information}.

\begin{table}[t]
\centering
\caption{Distributed classification results (10 seeds). $\B$: bit exchange (Mbit/epoch); $\Lambda$: capacity-weighted transmission cost (a.u.); Rate: nats/sample; Acc.: test accuracy (\%).}
\label{tab:results}
\scriptsize
\setlength{\tabcolsep}{2.2pt}
\resizebox{\columnwidth}{!}{%
\begin{tabular}{llccccc}
\toprule
Task & Scheme & $\B$ & $\Lambda$ & Rate & Acc.\ (\%) & NLL \\
\midrule
\multirow{5}{*}{T1}
 & Dense INL & 0.622 & 1085 & $20.7\pm9.7$ & $83.7\pm1.0$ & $0.381\pm0.031$ \\
 & Hop-count SPT & 0.138 & 384 & $26.9\pm10.5$ & $83.7\pm0.9$ & $0.374\pm0.022$ \\
 & Random tree & 0.198 & 303 & $57.2\pm19.5$ & $83.7\pm1.7$ & $0.389\pm0.045$ \\
 & D-INL & 0.184 & \textbf{77} & $85.0\pm22.1$ & $83.9\pm0.8$ & $0.407\pm0.063$ \\
 & D-INL+rate & 0.184 & \textbf{77} & $\mathbf{18.4\pm2.9}$ & $83.0\pm1.6$ & $0.402\pm0.037$ \\
\midrule
\multirow{5}{*}{T2}
 & Dense INL & 6.240 & 1825 & $39.7\pm12.8$ & $86.6\pm0.8$ & $0.301\pm0.014$ \\
 & Hop-count SPT & 0.960 & 804 & $39.7\pm13.2$ & $86.8\pm0.8$ & $0.302\pm0.016$ \\
 & Random tree & 1.344 & 333 & $95.3\pm39.2$ & $86.6\pm0.6$ & $0.310\pm0.017$ \\
 & D-INL & 1.440 & \textbf{218} & $115.8\pm38.5$ & $86.7\pm0.6$ & $0.309\pm0.013$ \\
 & D-INL+rate & 1.440 & \textbf{218} & $\mathbf{14.1\pm1.2}$ & $86.3\pm0.5$ & $0.309\pm0.010$ \\
\bottomrule
\end{tabular}%
}
\vspace{-1mm}
\end{table}

\subsection{Results}

Table~\ref{tab:results} and Fig.~\ref{fig:results} report the outcomes. Three observations:

\noindent\textbf{(i) Sparsity and accuracy:} In these controlled tasks, sparse training preserves predictive performance to within a few percentage points, with overlapping seed-wise variation. Relative to dense INL, D-INL+rate differs by $-0.7$ points on Task~1 and $-0.3$ points on Task~2. Pruning removes $70.4\%$ of candidate links on Task~1 ($8/27$) and $76.9\%$ on Task~2 ($15/65$); this does not imply no information loss on arbitrary data.

\noindent\textbf{(ii) Capacity awareness favors low-cost links, not bit count:} Hop-count SPT is a baseline for minimum hops and low raw exchange, not an algorithm whose bit optimum we seek to overturn. In Tasks 1--2 its direct sensor--fusion links are deliberately narrow: its capacity-weighted transmission cost $\Lambda$ is $5\times$ (Task~1) and $3.7\times$ (Task~2) higher than capacity-aware D-INL. By routing through high-capacity relays, D-INL reduces $\Lambda$ by $92.9\%$ (Task~1) and $88.1\%$ (Task~2) relative to dense INL—the largest reduction of any scheme. This reflects the controlled heterogeneous-link setting; the source-path objective differs from the measured union-edge proxy.

\noindent\textbf{(iii) The bottleneck controls a KL surrogate:} With matched stochastic encoders, regularization lowers the Gaussian KL proxy from $85.0$ to $18.4$ nats/sample ($78.3\%$) on Task~1 and from $115.8$ to $14.1$ ($87.8\%$) on Task~2. Mean accuracy changes by $-0.9$ and $-0.4$ points relative to unregularized D-INL, respectively. This does not reduce the fixed-width bit exchange in Table~\ref{tab:results}. So, a quantizer or entropy coder would be needed to turn the KL reduction into a measured payload saving.

\begin{figure}[t]
\centering
\subfloat[Acc.--exchange (T1).\label{fig:accbw}]{\includegraphics[width=0.49\linewidth]{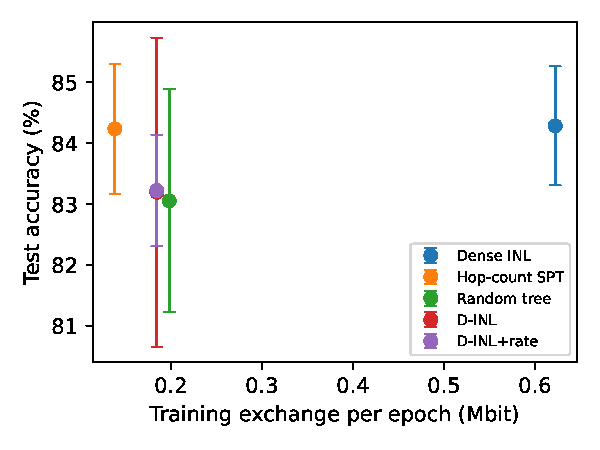}}\hfill
\subfloat[Rate--NLL (T1).\label{fig:raterisk}]{\includegraphics[width=0.49\linewidth]{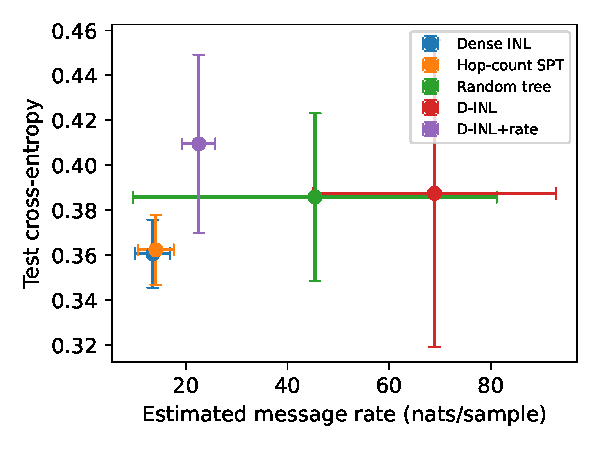}}
\caption{Task~1 under a matched stochastic-encoder protocol: accuracy versus nominal bit exchange, and Gaussian KL versus held-out NLL. Error bars show across-seed standard deviation.}
\label{fig:results}
\end{figure}

\subsection{Capacity-Heterogeneity Sensitivity}

To test whether the advantage is driven by narrow direct links, we independently draw all link capacities from $U[1,20]$ on the six-sensor graph and multiply only direct sensor--fusion capacities by $\kappa$. For each $\kappa$, 500 graph realizations compare fixed-width Dijkstra and hop-count trees; these are topology-cost tests, not additional accuracy experiments. Table~\ref{tab:cap-sweep} reports the median union-edge cost ratio, the fraction where Dijkstra has lower $\Lambda$, and the bit-exchange ratio. Equal capacity distributions ($\kappa=1$) reduce the advantage; for strong direct links ($\kappa=4$), the methods often coincide. The method is not uniformly superior, and hop-count never uses more edges in this study.

\begin{table}[t]
\centering
\caption{Independent-capacity sweep (500 graphs per setting). Ratios: D-INL/hop-count; wins: $\Pr(\Lambda_D<\Lambda_H)$.}
\label{tab:cap-sweep}
\scriptsize
\setlength{\tabcolsep}{4pt}
\begin{tabular}{cccc}
\toprule
$\kappa$ & Median $\Lambda_D/\Lambda_H$ & Wins (\%) & Median $\B_D/\B_H$ \\
\midrule
0.25 & 0.206 & 100.0 & 1.33 \\
0.50 & 0.426 & 99.4 & 1.33 \\
1.00 & 0.712 & 85.4 & 1.17 \\
2.00 & 0.983 & 52.2 & 1.17 \\
4.00 & 1.000 & 18.8 & 1.00 \\
\bottomrule
\end{tabular}
\end{table}

\subsection{Rate--Accuracy Trade-off Ablation}

To assess the KL--accuracy trade-off, we sweep $\lambda\in\{0,5\times10^{-4},10^{-3},2\times10^{-3},3\times10^{-3},4\times10^{-3},8\times10^{-3}\}$ on the Task~1 tree using the same stochastic training protocol (Fig.~\ref{fig:ablation}). The average Gaussian KL falls from $85.0$ to $8.8$ nats/sample (about $9.7\times$), while mean accuracy ranges from $83.0\%$ to $83.9\%$. NLL is not monotone and the KL proxy is not an implemented variable-length code; the ablation isolates the representation-information effect.

\begin{figure}[t]
\centering
\includegraphics[width=0.75\linewidth]{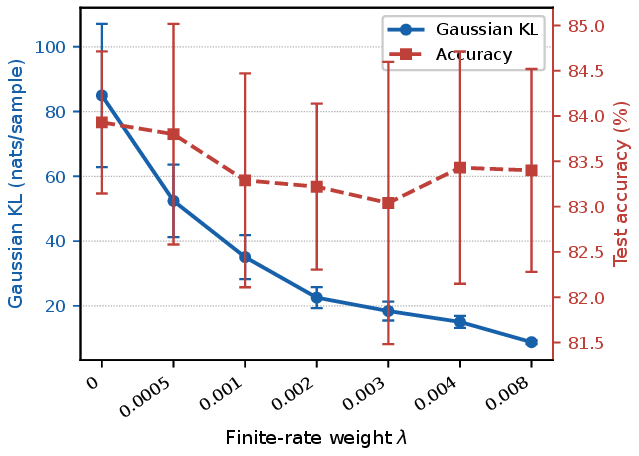}
\caption{Matched-stochastic ablation on Task~1. A larger $\lambda$ reduces the Gaussian KL proxy while mean test accuracy varies within approximately one percentage point. KL is not a measured wire rate.}
\label{fig:ablation}
\vspace{-2mm}
\end{figure}

\subsection{Distributed MNIST and MHEALTH}

The earlier distributed-MNIST task partitions each image into six distinct horizontal bands. Under sum aggregation, dense and hop-count INL reach $91.6\pm0.3\%$ and $91.7\pm0.3\%$ accuracy, while D-INL and D-INL+rate reach $88.2\pm0.7\%$ and $88.8\pm1.7\%$, respectively, at $\Lambda=77$ instead of $1085$ (dense) or $384$ (hop-count). Pruning therefore has an accuracy cost when sources observe complementary information. We exclude cross-protocol MNIST KL comparisons because the old unregularized encoders were trained deterministically.

We additionally evaluate \emph{MHEALTH}~\cite{banos2014mhealth}, which records activity from three physically separate wearables (chest, ankle, and arm) at 50 Hz. Each node forms six summary statistics per local channel on nonoverlapping, label-homogeneous 128-sample windows; we classify 12 activities. Subjects 1--7 are used for training, subject 8 for validation, and subjects 9--10 for testing, with train-only normalization. The 3-sensor, 2-relay graph has 11 candidate links with \emph{independent}, identically distributed $U[5,20]$ link capacities (no narrow-direct-link scaling); Table~\ref{tab:mhealth} reports five initialization seeds on a fixed subject split. We use $d=8$, sum aggregation, $\lambda=0.002$ for rate regularization, and 60-epoch early-stopped training. The centralized comparison uploads the same local summary features to a server in one direction, not raw sensor waveforms, and has no distributed backward messages.

\begin{table}[t]
\centering
\caption{MHEALTH, held-out subjects (five seeds). $\mathcal B$: Mbit per training epoch; $\Lambda$: capacity-weighted cost proxy; KL: nats/sample. Centralized feature upload is one-way and is not directly equivalent to distributed training.}
\label{tab:mhealth}
\scriptsize
\setlength{\tabcolsep}{2.3pt}
\resizebox{0.7\columnwidth}{!}{%
\begin{tabular}{lcccc}
\toprule
Scheme & $\mathcal B$ & $\Lambda$ & Acc. (\%) & KL \\
\midrule
Dense INL & 5.300 & 257.0 & $97.9\pm2.1$ & 912.7 \\
Hop-count SPT & 1.445 & 108.5 & $97.9\pm2.4$ & 1372.0 \\
Random tree & 2.120 & 90.8 & $94.5\pm5.2$ & 1907.9 \\
D-INL & 1.927 & 89.5 & $98.1\pm2.0$ & 1371.0 \\
D-INL+rate & 1.927 & 89.5 & $99.5\pm0.6$ & 32.8 \\
Centralized features & 4.155 & -- & $99.7\pm0.3$ & -- \\
\bottomrule
\end{tabular}}
\end{table}

On the held-out subjects, D-INL+rate obtains 99.5\% accuracy, compared with 97.9\% for hop-count and 99.7\% for centralized feature upload. It uses 33.3\% more nominal training bits than hop-count (1.927 versus 1.445 Mbit/epoch), but its cost proxy is 17.4\% smaller (89.5 versus 108.5). Some sampled random trees have a lower union-edge proxy (minimum 76.0) despite a higher mean (90.8), confirming that Dijkstra does not minimize that metric. These are real multi-device observations with emulated links, not a physical network deployment.

\subsection{Discussion}
\label{sec-dis-letter}

The numerical results illustrate a conditional design trade-off. Our aim is not to challenge hop-count routing on the number of hops or the nominal bit metric: if those are the sole objectives, hop-count is preferable. Instead, D-INL selects which neural dependencies exchange forward messages and reverse gradients under heterogeneous link capacities; it can use more bits yet yield a smaller capacity-weighted proxy. The independent-capacity sweep also shows a regime in which the advantage disappears. In contrast to conventional link routing, pruning determines which distributed neural dependencies and reverse-gradient paths exist. Gaussian KL control changes the latent statistical channel, not the number of fixed-width symbols on the wire. The MHEALTH experiment evaluates independently worn chest, wrist, and ankle sensors using an emulated edge graph; a physical gateway has not been deployed. The architecture is motivated by cases where uploading raw streams is undesirable or intermediate in-network fusion is required. If centralized collection is feasible, raw-data transmission followed by server-side training may be preferable; our experiments do not establish system-level superiority over that choice. We report no physical wearable-network testbed, measured latency, energy, or privacy gains. The single-parent tree does not constrain relay compute/queue load and offers no immediate backup path. Failure of a parent or its uplink disconnects descendants until re-routing. Load-aware admission and $k$-path redundancy would incur extra transmissions. These limitations, and the synthetic nature of link capacities, preclude claims of deployment-ready robustness.

\section{Conclusion}
\label{conclu}

This letter studied sparse INL training on a capacity-aware shortest-path tree and a Gaussian-KL-controlled message bottleneck. The routing result applies to sums of per-source path costs, whereas network-wide link exchange is measured separately. Across two simulated topologies, tree restriction suppresses up to $76.9\%$ of training links and lowers a capacity-weighted transmission-cost proxy by up to $92.9\%$ against dense INL. A randomized capacity study identifies settings where the comparison with hop-count reverses; distributed MNIST and held-out wearable subjects illustrate the cost of sparse fusion under complementary observations. The KL regularizer changes estimated message information, not fixed-width link bits. Load constraints, failover, hardware evaluation, and actual message coding remain future work.

\bibliographystyle{IEEEtran}
\bibliography{ref}

\end{document}